\documentclass[a4paper,12pt]{article} 
\usepackage[onehalfspacing]{setspace}
\usepackage{color}
\usepackage{enumitem}
\usepackage{amsmath,amssymb,amsthm}
\usepackage{subcaption}
\usepackage{graphicx}
\usepackage{algorithm}
\usepackage{algpseudocode}
\usepackage{color}
\usepackage{anysize}
\usepackage{multirow}
\usepackage{authblk}

\newtheorem{theorem}{Theorem}[section]
\newtheorem{lemma}[theorem]{Lemma}

\newtheorem{corollary}[theorem]{Corollary}

\theoremstyle{remark}

\title{Approximation Algorithms For The Euclidean Dispersion Problems}

\author{Pawan K. Mishra\thanks{Department of Computer Science and Engineering,
        Indian Institute of Technology Guwahati}, Gautam K. Das \thanks{Department of Mathematics,   Indian Institute of Technology Guwahati}}

\begin{document}

\maketitle

\begin{abstract}

	In this article, we consider the Euclidean dispersion problems. 
	Let $P=\{p_{1}, p_{2}, \ldots, p_{n}\}$ be a set of $n$ points in $\mathbb{R}^2$. For each point $p \in P$ and $S \subseteq P$, we define $cost_{\gamma}(p,S)$ as the sum of Euclidean distance from $p$ to the nearest $\gamma $ point in $S \setminus \{p\}$. We define $cost_{\gamma}(S)=\min_{p \in S}\{cost_{\gamma}(p,S)\}$ for  $S \subseteq P$. In the $\gamma$-dispersion problem, a set $P$ of $n$ points in $\mathbb{R}^2$ and a positive integer $k \in [\gamma+1,n]$ are given. The objective is to find a subset $S\subseteq P$ of size $k$ such that $cost_{\gamma}(S)$ is maximized. We consider both $2$-dispersion and $1$-dispersion problem in $\mathbb{R}^2$. Along with these, we also consider $2$-dispersion problem when points are placed on a line.

	In this paper, we propose a simple polynomial time $(2\sqrt 3 + \epsilon )$-factor approximation algorithm for the $2$-dispersion problem, for any $\epsilon > 0$, which is an improvement over the best known approximation factor $4\sqrt3$ [Amano, K. and Nakano, S. I., An approximation algorithm for the $2$-dispersion problem, IEICE Transactions on Information and Systems, Vol. 103(3), pp. 506-508, 2020]. Next, we develop a common framework for designing an approximation algorithm for the Euclidean dispersion problem. With this common framework, we improve the approximation factor  to $2\sqrt 3$ for the $2$-dispersion problem in $\mathbb{R}^2$.  Using the same framework, we propose a polynomial time algorithm, which returns an optimal solution  for the $2$-dispersion problem when  points are  placed on a line. Moreover, to show the effectiveness of the  framework, we  also propose  a $2$-factor approximation algorithm for the $1$-dispersion problem in $\mathbb{R}^2$.

\end{abstract}

\section{Introduction} \label{sec:intro}

The facility location problem is one of the extensively studied optimization problems.   Here, we are given a set of locations on which facilities can be placed and a positive integer $k$, and the goal is to place $k$ facilities on those locations so that a specific objective is satisfied. For example, the objective is to  place these facilities such that their closeness is undesirable.  Often, this closeness measured as a function of the distances between a pair of facilities. We refer to such facility location problem as a dispersion problem. More specifically, we wish to minimize the interference between the placed facilities. The most studied dispersion problem is the \emph{max-min dispersion problem}.

In the \emph{max-min dispersion problem},  we are given a set $ P = \{p_{1}, p_{2}, \ldots, p_{n}\}$ of $n$ locations,  the non-negative distances between each pair of locations $p,q \in P$, and a positive integer $k$ ($k \leq n$). Here, $k$ refers to the number of facilities to be opened and distances are assumed to be symmetric. The objective is to find a $k$ size subset $S \subseteq P$  of locations such that $cost(S)= \min \{d(p,q) \mid p,q \in S \}$ is maximized, where $d(p,q)$ denotes the distance between $p$ and $q$.  This problem is known as $1$-dispersion problem in the literature. 
%
In this article, we consider a variant of the max-min dispersion problem. We refer to it as a $2$-dispersion problem. Now, we define $2$-dispersion problem as follows: 

\textbf{\emph{2-dispersion problem:}}  \emph{Let $P=\{p_{1}, p_{2}, \ldots, p_{n}\}$ be a set of $n$ points in $\mathbb{R}^2$. For each point $p \in P$ and $S \subseteq P$, we define $cost_{2}(p,S)$ as the sum of Euclidean distance from $p$ to the first closest point in $S \setminus \{p\}$ and the second closest point in $S\setminus \{p\}$. We also define $cost_{2}(S)=\min_{p \in S}\{cost_{2}(p,S)\}$ for each $S \subseteq P$. In the $2$-dispersion problem, a set $P$ of $n$ points in $\mathbb{R}^2$ and a positive integer $k \in [3,n]$ are given. The objective is to find a subset $S\subseteq P$ of size $k$ such that $cost_{2}(S)$ is maximized. }

We find an immense number of applications for the dispersion problem in the real world. The situation in which we want to open chain stores in a community has generated our interest in the dispersion issue. In order to eliminate/prevent self-competition, we need to open stores far away from each other. Another situation in which the issue of dispersion occurs is  installing  hazardous structures, such as nuclear power plants and oil tanks. These facilities need to be dispersed to the fullest degree possible so that an accident at one of the facilities would not affect  others. The dispersion problem also has its application in  information retrieval where we need to find a small subset of data with some desired variety from an extensive data set such that a small subset is a reasonable sample to overview the large data set. 

\section{Related Work} \label{sec: relatedwork}

In 1977, Shier \cite {shier1977min} in his study on the $k$-center problem on a tree, studied the max-min dispersion problem on trees. In 1981, Chandrasekharan and Daughety \cite{chandrasekaran} studied  max-min dispersion problem on a tree network. The max-min dispersion problem is NP-hard even when the distance function satisfies  triangular  inequality \cite{erkut1990}. Wang and Kuo \cite{wang1988study} introduced the geometric version of the max-min dispersion problem. They consider the problem in the $d$-dimensional space where the distance between two points is Euclidean. They proposed a dynamic programming that solves the problem for $d=1$ in $O(kn)$ time.  They also proved that the problem is NP-hard for $d=2$. 
In \cite{white2}, White studied  the max-min dispersion problem and proposed  a  $3$-factor approximation result. Later in 1994, Ravi et al. \cite{ravi}  also studied max-min dispersion problem where they  proposed a $2$-factor approximation algorithm when the distance function satisfies  triangular inequality. Moreover, they showed that when distance satisfies the triangular inequality, the problem cannot be approximated within the factor of $2$ unless $P=NP$. 

Recently, in \cite{akagi}, the exact algorithm for the problem was shown by establishing a  relationship between the max-min dispersion problem and the maximum independent set problem. They proposed an $O(n^{wk/3} \log n)$ time where $w < 2.373$. In \cite{akagi}, Akagi et al. also studied two special cases where a set of $n$ points lies on a line and  a set of $n$ points lies on a circle separately. They proposed a polynomial time exact algorithm for both special cases. 

The other popular variant of the dispersion problem is \emph{max-sum $k$-dispersion problem}. Here, the objective is to maximize the sum of distances between $k$ facilities. Erkut \cite{erkut1990} idea's can be  adapted to show that the problem is NP-hard. Ravi et al. \cite{ravi} gave a polynomial time exact algorithm when the points are placed on a line. They also proposed a $4$-factor approximation algorithm if the distance function satisfies triangular inequality. In \cite{ravi}, they also proposed a $(1.571+ \epsilon)$-factor approximation algorithm for $2$-dimensional Euclidean space, where $\epsilon > 0$. In \cite{birnbaum} and \cite{hassin}, the approximation factor of $4$ was improved to $2$.  One can see \cite{baur2001approximation} and \cite{chandra} for other variations of the dispersion problems.
In comparison with max-min dispersion ($1$-dispersion) problem, a handful amount of research has been done in $2$-dispersion problem.  Recently, in 2018, Amano and Nakano \cite{amano} proposed a greedy algorithm, which produces an $8$- factor approximation result.  In 2020, \cite{amano2020} they analyzed the same greedy algorithm proposed in \cite{amano} and proposed a $4\sqrt{3}(\approx 6.92)$-factor approximation result.

\subsection{Our Contribution}
In this article, we first consider the $2$-dispersion problem in $\mathbb{R}^2$ and  propose a simple polynomial time  $(2\sqrt3+ \epsilon)$-factor approximation algorithm for any $\epsilon >0$. The best known result in the literature is $4\sqrt3$-factor approximation algorithm\cite{amano2020}. We also develop a common framework that improves the approximation factor to $2\sqrt3$ for the same problem.  We present a polynomial time optimal algorithm for $2$-dispersion problem if the input points lies on a line. Though a $2$-factor approximation algorithm available in the literature for the $1$-dispersion problem in $\mathbb{R}^2$\cite{ravi}, but to show the effectiveness of the proposed common framework, we  propose a $2$-factor approximation algorithm for the  $1$-dispersion problem using the developed framework. 

\subsection{Organization of the Paper}
The remainder of the paper is organized as follows. In Section \ref{section3}, we propose a $(2\sqrt3 +\epsilon)$-factor approximation algorithm for the $2$-dispersion problem in $\mathbb{R}^2$, where $\epsilon > 0$. In Section \ref{section4}, we propose a common framework for the dispersion problem.  Using the framework, followed by   $2\sqrt3$-factor approximation result  for the $2$-dispersion problem in  $\mathbb{R}^2$, a  polynomial time optimal algorithm for the $2$-dispersion problem on a line and  $2$-factor approximation result for the $1$-dispersion problem in $\mathbb{R}^2$. Finally, we conclude the paper in Section \ref{section5}.

\section{$(2\sqrt3 +\epsilon)$-Factor Approximation Algorithm}\label{section3}
In this section, we propose a $(2\sqrt3 +\epsilon)$-factor approximation algorithm for the $2$-dispersion problem, for any $\epsilon > 0$. Actually, we consider the same algorithm proposed in \cite{amano2020}, but using different argument, we will show that for any $\epsilon > 0$, it is a $(2\sqrt3+\epsilon)$-factor approximation algorithm. For completeness of this article, we prefer to discuss the algorithm briefly as follows. Let $I=(P, k)$ be an arbitrary instance of the $2$-dispersion problem, where $P=\{p_{1}, p_{2}, \ldots p_{n}\}$  is  the set of $n$ points in $\mathbb{R}^2$ and $k \in [3,n]$ is a positive integer. Initially, we choose a subset $S_{3} \subseteq S$ of size $3$ such that $cost_{2}(S_{3})$ is maximized. Next, we add one point  $p \in P$ into $S_{3}$ to construct $S_{4}$, i.e., $S_{4}=S_{3} \cup \{p\}$, such that $cost_{2}(S_4)$ is maximized and continues this process up to the construction of $S_{k}$. The pseudo code of the algorithm is described in Algorithm \ref{GreedyDispersionAlgorithm}. 
\begin{algorithm}[h]
	
	\caption{GreedyDispersionAlgorithm$(P, k)$}
	\textbf{Input: } A set $P=\{p_1, p_2, \ldots, p_n\}$ of $n$ points, and a positive integer  $k (3 \leq k \leq n)$.
	
	\textbf{Output:} A subset $S_k \subseteq P$ of size $k$.	
	\begin{algorithmic}[1]
		\State Compute $\{p_{i_1}, p_{i_2}, p_{i_3}\} \subseteq P$ such that $cost_2(S_3)$ is maximized. \label{line1algo1}
		\State $S_3 = \{p_{i_1}, p_{i_2}, p_{i_3}\}$
		\For {($j = 4, 5 \ldots, k$)} \label{line2algo}
		\State Let $p \in P\setminus S_{j-1}$ such that $cost_2(S_{j-1} \cup \{p\})$ is maximized. \label{greedy1}
		\State $S_j \leftarrow S_{j-1} \cup \{p\}$
		\EndFor
		\State return $(S_k)$
	\end{algorithmic}  \label{GreedyDispersionAlgorithm}
\end{algorithm}

\begin{theorem}
	For any $\epsilon > 0$, Algorithm \ref{GreedyDispersionAlgorithm} produces $(2\sqrt{3}+\epsilon)$-factor approximation result in polynomial time.
\end{theorem}

\begin{proof}
	Let $I = (P, k)$ be an arbitrary input instance of the $2$-dispersion problem, where $P=\{p_1, p_2, \ldots, p_n\}$ is the set of $n$ points  and $k$ is a positive integer. Let $S_k$ and $OPT$ be the output of Algorithm \ref{GreedyDispersionAlgorithm} and optimum solution, respectively, for the instance $I$. To prove the theorem, we have to show that $\frac{cost_2(OPT)}{cost_2(S_k)} \leq 2\sqrt3+\epsilon$. Here we use induction to show that $cost_2(S_i) \geq \frac{cost_2(OPT)}{2\sqrt3+\epsilon}$ for each $i = 3, 4, \ldots, k$. Since $S_3$ is an optimum solution for 3 points (see line number \ref{line1algo1} of Algorithm \ref{GreedyDispersionAlgorithm}), therefore $cost_2(S_3) \geq cost_{2}(OPT) \geq  \frac{cost_2(OPT)}{2\sqrt3+\epsilon}$ holds. Now, assume that the condition holds for  each $i$ such that $3 \leq i < k$. We will prove that the condition holds for $(i+1)$ too.
	
	   Now, we define a disk $ D_i$ centered at each $p_i \in P$ as follows: $ D_i = \{p_\ell \in \mathbb{R}^2|d(p_{i}, p_{\ell}) \leq \frac{cost_2(OPT)}{2\sqrt{3}+\epsilon}\}$. Let $D^*$ be a set of disks corresponding  to  each point  in $OPT$. A point $p_{j}$ is \emph{contained} in $D_{i}$, if $d(p_{i},p_{j})\leq \frac{cost_2(OPT)}{2\sqrt{3}+\epsilon}$.
	\begin{lemma} \label{lemma01}
		For any point $p_{i} \in P$, $|D_{i} \cap OPT | \leq 2$. 
	\end{lemma} 
	\begin{proof}
		On the  contrary assume that three points $p_{a}, p_{b}, p_{c} \in D_{i} \cap OPT$. Let $S=\{p_{a}, p_{b}, p_{c}\}$. Without loss of generality assume that $cost_{2}(p_{a}, S) \leq cost_{2}(p_{b}, S)$ and $cost_{2}(p_{a}, S) \leq cost_{2}(p_{c}, S)$, i.e., $d(p_{a}, p_{b})+ d(p_{a}, p_{c}) \leq d(p_{a}, p_{b})+ d(p_{b}, p_{c})$ and $d(p_{a}, p_{b})+ d(p_{a}, p_{c}) \leq d(p_{a}, p_{c})+ d(p_{b}, p_{c})$, which leads to $d(p_{a}, p_{b}) \leq d(p_{b}, p_{c})$ and $d(p_{a}, p_{c}) \leq d(p_{b}, p_{c})$. We notice that maximizing $d(p_{a},p_{b})+d(p_{a},p_{c})$ results in minimizing $d(p_{b},p_{c})$(see Figure \ref{figure1}). The minimum value of $d(p_{b}, p_{c})$ is $\sqrt3 \frac{cost_2(OPT)}{2\sqrt{3}+\epsilon}$ as both $d(p_{a}, p_{b})$ and $d(p_{a}, p_{c})$ is less than equal to $d(p_{b}, p_{c})$. Therefore, from the packing argument inside a disk,  $d(p_{a},p_{b})+d(p_{a},p_{c})$ is maximum 
		if $p_{a},p_{b}, p_{c}$ are on an equilateral triangle and on the boundary of the disk $D_{i}$. Then, $cost_{2}(S) \leq d(p_{a},p_{b})+d(p_{a},p_{c}) \leq \sqrt 3 \frac{cost_2(OPT)}{2\sqrt{3}+\epsilon}+\sqrt 3 \frac{cost_2(OPT)}{2\sqrt{3}+\epsilon}=2\sqrt3 \frac{cost_2(OPT)}{2\sqrt{3}+\epsilon} < cost_{2}(OPT)$, which leads to a contradiction to the optimal value $cost_{2}(OPT)$. Therefore for any $p_i \in P$, $D_i$ contains at most two points from the optimal set $OPT$.
		
		\begin{figure}[h!]
			\centering
			\includegraphics[scale=1]{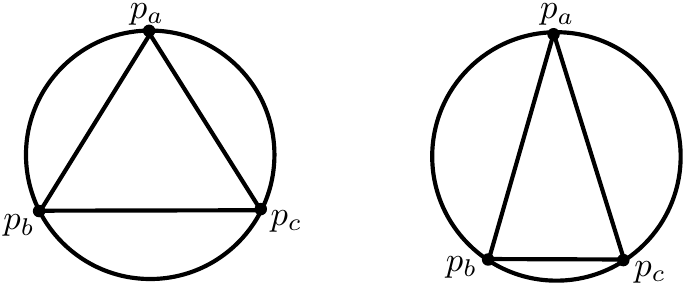}
			\caption{Points $p_{a}, p_{b}, p_{c} \in D_{i} $} \label{figure1}

		\end{figure}
	\end{proof}
%
	
	\begin{lemma} \label{lemma02}
		For some   $p_{j}  \in OPT$,	$| D_j \cap S_i| < 2$. 
	\end{lemma}
	\begin{proof}
		On the contrary assume that there does not exist any $j \in [1,k]$ such that $| D_j \cap S_i| < 2$. Let $ D^*=\{D_{i} \mid p_{i} \in OPT\}$.
		Construct a bipartite graph $H(S_i \cup  D^*, {\cal E})$ as follows: (i) $S_i$ and $ D^* = \{ D_1,  D_2, \ldots, D_k\}$ are two partite vertex sets, and (ii) for $u_{i}' \in S_{i}$, $(u_i', D_j) \in {\cal E}$ if and only if $u_i'$ is contained in $ D_j$. According to assumption, each disk $D_j$ contains at least 2 points from $S_i$. Therefore, the total degree of the vertices in $ D^*$ in $H$ is at least $2k$. Note that $| D^*| = k$. On the other hand, the total degree of the vertices in $S_i$ in   $H$ is at most $2 \times |S_i|$ (see Lemma \ref{lemma01}). Since $|S_i| < k$ (based on the assumption of the induction hypothesis), the total degree of the vertices in $S_i$ in   $H$ is less than $2k$, which leads to a contradiction that the total degree of the vertices in $D^*$ in  $H$ is at least $2k$. Thus, there exist at least one $p_{j} \in OPT$ such that $|D_{j} \cap S_i| < 2$. 
	\end{proof}
	Without loss of generality,  assume that  disk $ D_j \in  D^*$ has at most one point from the set $S_{i}$. Suppose $ D_j$ contains only one point of the set $S_{i}$, then the distance of $p_{j}$ to the second closest point in $S_{i}$ is greater than $ \frac{cost_{2}(OPT)}{2\sqrt3 +\epsilon}$(see Figure \ref{figure2} ). Also, from triangular inequality $d(p_{i}, p_{j}) +d(p_{i}, p_{\ell}) > \frac{cost_{2}(OPT)}{2\sqrt3 +\epsilon}$ for each point $p_{\ell} \in S_{i}$. So, we can add the point $p_{j} \in OPT$ to the set $S_{i}$ to construct set $S_{i+1}$. Here, $S_{i+1}=S_{i} \cup \{p_{j}\}$. Therefore, the cost of $S_{i+1} \geq \frac{cost_{2}(OPT)}{2\sqrt3 +\epsilon}$. 
	
	Now,  assume that $D_j$ does not contain any point from the set $S_{i}$, then the distance of the point $p_{j} \in OPT $ to any point of $S_{i}$ is greater than $ \frac{cost_{2}(OPT)}{2\sqrt3 +\epsilon}$. By adding the point $p_{j}$ in the set $S_{i}$, we construct the set $S_{i+1}$, which leads to the $cost_{2}(S_{i+1})\geq \frac{cost_{2}(OPT)}{2\sqrt3 +\epsilon}$. 
		Since our algorithm chooses a point (see line number \ref{greedy1} of Algorithm \ref{GreedyDispersionAlgorithm}) that maximizes $cost_{2}(S_{i+1})$, therefore algorithm will always choose a point in the iteration $i+1$ such that  $cost_{2}(S_{i+1}) \geq \frac{cost_{2}(OPT)}{2\sqrt3 +\epsilon }$. 
		
	By the help of Lemma \ref{lemma01}  and Lemma \ref{lemma02}, we can conclude that the $cost_{2}(S_{i+1})\geq \frac{cost_{2}(OPT)}{2\sqrt3 +\epsilon}$ and thus condition holds for $(i+1)$ too.  
	
	Therefore, for any $\epsilon > 0$, Algorithm \ref{GreedyDispersionAlgorithm} produces $(2\sqrt{3}+\epsilon)$-factor approximation result in polynomial time.
	\begin{figure}[h!]
		\centering
		\includegraphics[scale=1]{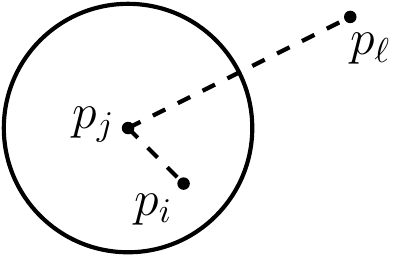}
		\caption{Points $p_{j}, p_{i} \in D_{j}$ and $p_{\ell}$ outside the disk $D_{j}$ }  \label{figure2} 
	\end{figure}
\end{proof}
\section{ An  Algorithm for the Dispersion Problem} \label{section4}
In this section, we propose an algorithm  for the dispersion problem. 
It is a common algorithm for $1$-dispersion, $2$-dispersion problem in $\mathbb{R}^2$ and $1$-dispersion/$2$-dispersion problem in $\mathbb{R}$. Input of the algorithm are (1) a set $P=\{p_1, p_2, \ldots p_n\}$ of $n$ points, (2) an integer $\gamma(=1$or $2)$ for the $\gamma$-dispersion problem, and (3) an integer $k (\gamma+1\leq k\leq n)$. In the first line of the algorithm, we set the value of a constant $\lambda$. If $\gamma=2 $ and points are in $\mathbb{R}^2$  (resp. $\mathbb{R}$), then we set $\lambda=2\sqrt3$  (resp. $\lambda=1)$, and if $\gamma=1$  and points are in $\mathbb{R}^2$, then we set $\lambda=2$. We prove that the algorithm is $\lambda$-factor approximation algorithm. We use $S_{i} (\subseteq P)$ to denote a set of size $i$.
We start algorithm with $S_{\gamma+1} \subseteq P$ containing $\gamma+1$ points as a solution set. Next, iteratively we add one by one point from $P$ into the solution set to get a final solution set, i.e., if we have a solution set $S_{i}$ of size $i$, then we add one more point into $S_{i}$ to get solution set $S_{i+1}$ of size  $i+1$. Let  $\alpha=cost_{\gamma}(S_{i})$. Now, we add a point from $P\setminus S_{i}$ into $S_{i}$ to get $S_{i+1}$ such that $ cost_{\gamma}(S_{i+1}) \geq \frac{\alpha}{\lambda} $. We stop this iterative method if we have $S_{k}$ or no more point addition 
is possible.  We repeat the above process  for each distinct $S_{\gamma+1} \subseteq P$ and report the solution for which the $\gamma$-dispersion cost value is maximum.

\begin{algorithm}[h]
	\caption{Dispersion Algorithm$(P, k, \gamma)$}
	\textbf{Input: }A set $P$ of $n$ points, a positive integer $\gamma$ and an integer $k$ such that $\gamma+1 \leq k \leq n$ \\
	\textbf{Output: }A subset $S_{k} \subseteq P$ such that $|S_{k}|=k$ and $\beta=cost_{\gamma}(S_{k})$.	
	\begin{algorithmic}[1]
		\State If $\gamma=2$ (resp. $\gamma=1$), then  $\lambda \gets 2\sqrt3$   (resp. $\lambda \gets 2$), and if points are on a line then  $\lambda \gets 1$  \label{line1}.
		\State $\beta \gets 0$ // Initially, $cost_{\gamma}(S_{k})=0$
		\For {each subset $S_{\gamma+1} \subseteq P $ consisting of $\gamma +1$ points } \label{line2algo}
		
		\State  Set  { $ \alpha  \gets cost_{\gamma}(S_{\gamma+1})$}
		\State  Set   { $\rho \gets \alpha / \lambda$ }
		\If{ $\rho > \beta$}
		\State $flag \gets 1$,  $i \gets \gamma+1$
		
		\While{$ i < k$  and $flag  \neq 0$}
		\State $flag  \gets 0$
		\State choose a point $p \in P\setminus S_{i}$ (if possible) such that $cost_{\gamma} (S_{i} \cup \{p\}) \geq \rho$ and $cost_{\gamma}(p,S_{i})=\min_{q \in P\setminus S_{i}}cost_{\gamma}(q,S_{i})$. \label{alphaline}
		\If {such point $p$ exists in step $10$}
		\State $ S_{i+1} \gets S_{i} \cup \{p\} $ \label{line10algo}
		\State $i \gets i+1$, $flag \gets 1$
		\EndIf
		\EndWhile
		\If {$i=k$} 
		\State $S_{k} \gets S_{i}$ and  $\beta \gets \rho$ \label{line15algo}
		\EndIf 
		\EndIf
		\EndFor
		\State return $(S_{k},\beta)$
	\end{algorithmic}  \label{algo}
\end{algorithm}
\subsection{$2\sqrt3$-Factor Approximation Result for the $2$-Dispersion Problem}
Let $S^{*} \subseteq P=\{p_{1},p_{2},\ldots, p_{n}\} $  be an optimal solution for a given instance $(P,k)$ of the  $2$-dispersion problem  and  $S_{k} \subseteq P$ be a solution returned by  greedy Algorithm \ref{algo} for the given instance, provided $\gamma=1$ as an additional input. A point $s_{o}^{*} \in S^{*}$ is said to be a solution point if $cost_{2}(S^{*})$ is defined by $s_{o}^{*}$, i.e.,  $cost_{2}(S^{*}) = d(s_{o}^{*}, s_{r}^{*}) + d(s_{o}^{*}, s_{t}^{*}) $ such that  (i) $s_{r}^{*}, s_{t}^{*} \in S^{*}$, and (ii) $s_{r}^{*}$ and $s_{t}^{*}$ are the first and  second closest points  of $s_{o}^{*}$ in $S^{*}$, respectively. We call $s_{r}^{*}$, $s_{t}^{*}$ as supporting  points. Let $\alpha=cost_{2}(S^*)$. In this problem, the value of $\lambda$ is $2\sqrt{3}$ (line number \ref{line1} of Algorithm \ref{algo}).

\begin{lemma}
	The triangle formed by three points $s_{o}^{*}$, $s_{r}^{*}$  and $s_{t}^{*}$ does not contain any point in  $S^{*} \setminus \{s_{o}^*,s_{r}^*,s_{t}^*\}$, where $s_{o}^*$ is the solution  point, and $s_{r}^{*}$,  $s_{t}^{*}$ are supporting  points.
\end{lemma}
\begin{proof}
	Suppose there exist a point $s_{m}^{*} \in S^{*}$  inside the triangle formed by $s_{o}^{*}$, $s_{r}^{*}$ and $s_{t}^{*}$. Now, if $d(s_{o}^{*}, s_{r}^{*}) \geq d(s_{o}^{*}, s_{t}^{*})$ then
	$d(s_{o}^{*}, s_{t}^{*}) + d(s_{o}^{*}, s_{m}^{*}) < d(s_{o}^{*}, s_{r}^{*}) + d(s_{o}^{*}, s_{t}^{*})$ which contradict the optimality of $cost_{2}(S^{*})$.  A similar argument will also  work for  $  d(s_{o}^{*}, s_{r}^{*}) < d(s_{o}^{*}, s_{t}^{*})  $.
\end{proof}

In this problem, $\rho=\frac{\alpha}{\lambda}=\frac{cost_{2}(S^*)}{2\sqrt3}$. We define a disk  $D_{i}$ centered at  $p_{i} \in P$ as follows: $D_{i} = \{p_j \in \mathbb{R}^2|d(p_{i}, p_{j}) \leq \rho\}$.  Let $D=\{D_{i}\mid p_{i} \in P\}$. Let $D^*$ be the subsets of $D$ corresponding  to disks centered at  points in $S^*$. A point $p_{j}$ is \emph{properly contained} in $D_{i}$, if $d(p_{i},p_{j}) <  \rho $, whereas  if $d(p_{i},p_{j}) \leq  \rho $, then we say that point $p_{j}$ is \emph{contained} in $D_{i}$.

\begin{lemma} \label{lemma2}
	For any point $p \in P $, if $D^p=\{q \in \mathbb{R}^2 \mid d(p,q) \leq \rho\}$ then $D^p$ properly contains at most two points of the optimal set $S^*$.
\end{lemma}
\begin{proof}
	On the contrary assume that three points $p_{a},p_{b}, p_{c} \in S^*$ such that $p_{a},p_{b}, p_{c}$ are properly contained in $D^p$. Using the similar arguments discussed in the proof of Lemma \ref{lemma01}, $cost_{2}(\{p_{a},p_{b}, p_{c}\})$ 
	is maximum if $p_{a},p_{b}, p_{c}$ are on equilateral triangle inside $D^p$.   Therefore, $d(p_{a},p_{b})=d(p_{a},p_{c})=d(p_{b},p_{c})$. Now, $cost_{2}(\{p_{a},p_{b}, p_{c}\})= d(p_{a},p_{b})+d(p_{a},p_{c}) < \sqrt 3 \rho+\sqrt 3 \rho=2\sqrt3 \rho= cost_{2}(S^{*})$. Therefore, $p_{a},p_{b}, p_{c} \in S^*$ and $cost_{2}(\{p_{a},p_{b}, p_{c}\}) < cost_{2}(S^{*})$ leads to a contradiction. Thus, the lemma.

\end{proof}	
\begin{lemma}For any  three points $\{p_{a}, p_{b}, p_{c}\} \in S^*$, there does not exist any point  $s \in \mathbb{R}^2$ such that $s$ is properly contained in  $D_a \cap D_b \cap D_c$. \label{lemma3}
\end{lemma}
\begin{proof}
	On the contrary assume that $s$ is properly contained in   $D_{a} \cap D_{b} \cap D_{c} $. This implies  $d(p_{a}, s) < \rho$,
	$d(p_{b}, s) <  \rho$ and $d(p_{c}, s) < \rho $. Therefore, the disk $D^s =\{q \in \mathbb{R}^2 \mid d(s,q) \leq \rho \}$ properly contains three points $p_{a}, p_{b}$ and $ p_{c}$, which is a contradiction to Lemma \ref{lemma2}. Thus, the lemma.  
	
\end{proof}

\begin{corollary} \label{cor1}
	For any point $p \in P $, if $D' \subseteq D^*$ is the subset of disks that contains  $p$, then $|D'| \leq 3$ and $p$ lies on the boundary of each disk in $D'$.  
\end{corollary} 
\begin{proof}
	Follows from   Lemma \ref{lemma3}.  
\end{proof}
\begin{corollary} \label{cor2}
	For any point $p \in P $, if $D'' \subseteq D^*$ is the  subset of disks that properly contains  point $p$, then $|D''| \leq  2$.
\end{corollary}
\begin{proof}
	Follows from Lemma \ref{lemma3} and Corollary \ref{cor1}. 
\end{proof}

\begin{lemma} \label{lemma010}
	Let $S\subseteq P$ be a set of points such that $|S| < k$. If $cost_{2}(S)\geq \rho$, then there exists at least one disk $D_{j} \in D^*=\{D_{1},D_{2}, \ldots, D_{k}\}$ that properly contains at most one point from the set $S$.
\end{lemma}
\begin{proof}
	On the contrary assume that each   $D_{j} \in D^*$ properly contains at least two points from the set $S$. Construct a bipartite graph $G(S \cup  D^*, {\cal E})$ as follows: (i) $S$ and $D^*$ are two partite vertex sets, and (ii) for $u \in S$, $(u,  D_j) \in {\cal E}$ if and only if $u$  is properly contained in $ D_j$. According to assumption, each disk $ D_j$ contains at least $2$ points from the set $S$. Therefore, the total degree of the vertices in $ D^*$ in $G$ is at least $2k$. Note that $| D^*| = k$. On the other hand, the total degree of the vertices in $S$ in   $G$ is at most $2 \times |S|$ (see Corollary  \ref{cor2}). Since $|S| < k$, the total degree of the vertices in $S$ in   $G$ is less than $2k$, which leads to a contradiction that the total degree of the vertices in $ D^*$ in   $G$ is at least $2k$. Thus, there exist at least one disk $D_{j} \in D^*$ such that the disk $D_{j}$ properly contains at most one point from the set $S$.

\end{proof}

\begin{theorem}
	Algorithm \ref{algo} produces a $2\sqrt{3}$-factor approximation result for the $2$-dispersion problem in $\mathbb{R}^2$. \label{thm1}
\end{theorem}

\begin{proof}

	Since it is a $2$-dispersion problem, so $\gamma=2$ and set $\lambda=2\sqrt3$ in line number \ref{line1} of  Algorithm \ref{algo}. Now, assume $\alpha=cost_{2}(S^*)$ and $\rho=\frac{\alpha}{\gamma}=\frac{cost_{2}(S^*)}{2\sqrt3}$, where $S^*$ is an optimum solution.
	Here, we show that   Algorithm \ref{algo} returns a solution  set $S_{k}$ of size $k$ such that $cost_{2}(S_{k}) \geq \rho=\frac{cost_{2}(S^*)}{2\sqrt3}$. More precisely, we show that Algorithm \ref{algo} returns a solution $S_{k}$ of size $k$ such that  $cost_{2}(S_{k})\geq \frac{cost_{2}(S^*)}{2\sqrt3}$ and $S_{k} \supseteq \{{s_{o}^{*},s_{r}^{*},s_{t}^{*}}\}$, where $s_{o}^*$ is the solution point and $s_{r}^*$ and $s_{t}^*$ are supporting points, i.e., $cost_{2}(S^*)=d(s_{o}^*,s_{r}^*)+d(s_{o}^*, s_{t}^*)$. Now, consider the case when $S_{3}=\{ s_{o}^{*},s_{r}^{*}, s_{t}^{*}\}$ in  line number \ref{line2algo} of  Algorithm \ref{algo}. Our objective is to show that if $S_{3}=\{ s_{o}^{*},s_{r}^{*}, s_{t}^{*}\}$  in  line number \ref{line2algo}  of  Algorithm \ref{algo}, then it computes a solution $S_{k}$ of size $k$ such that $cost_{2}(S_{k}) \geq \frac{cost_{2}(S^*)}{2\sqrt3}$. Note that  any other solution returned by  Algorithm \ref{algo} has a $2$-dispersion cost better than $\frac{cost_{2}(S^*)}{2\sqrt3}$. Therefore, it is sufficient to prove that if  $S_{3}=\{ s_{o}^{*},s_{r}^{*}, s_{t}^{*}\}$ in line number \ref{line2algo} of  Algorithm \ref{algo}, then the size of  $S_{k}$ (updated) in line number \ref{line15algo}  of  Algorithm \ref{algo} is $k$ as every time Algorithm \ref{algo} added a point (see line number \ref{line10algo}) into the set  with the property that $2$-dispersion  cost of the updated set is greater than or equal to  $ \frac{cost_{2}(S^*)}{2\sqrt3}$. Therefore,  we consider $S_{3}=\{ s_{o}^{*},s_{r}^{*}, s_{t}^{*}\}$ in  line number \ref{line2algo}  of  Algorithm \ref{algo}.

	We use induction to  establish the condition    $cost_{2}(S_{i}) \geq \rho$  for each $i=3,4,\ldots k$. Since $S_{3}=S_{3}^*$, therefore $cost_{2}(S_{3})=cost_{2}(S_{3}^*)=\alpha > \rho $ holds. Now, assume that the condition $cost_{2}(S_{i})\geq \rho$ holds for each  $i$ such that $3 \leq i < k$. We will prove that the condition $cost_{2}(S_{i+1})\geq \rho$ holds for $(i+1)$ too.

	Let $D^*$ be the set of disks centered at the points in $S^*$ such that the radius of each disk is $\rho$. Since $i<k$ and $S_{i}\subseteq P$ with condition $cost_{2}(S_{i})\geq \rho$, there exist at least one disk, say $D_{j} \in D^*$ that properly contains at most one point in $S_{i}$ (see Lemma \ref{lemma010}). We will show that $cost_{2}(S_{i+1})=cost_{2}(S_{i} \cup \{p_{j}\}) \geq \rho$, where $p_{j}$ is the center of the disk $D_{j}$.
	Suppose, $ D_j$ contains only one point $ p_{x} \in S_i$, then $p_{x}$ is the first closest point of $p_{j}$ in the set $S_{i}$. Now, by Corollary \ref{cor1} and by Lemma \ref{lemma010}, we claim the second closest point $p_{\ell}$ of $p_{j}$ in  the set $S_{i}$ may lie (1)  on the boundary of the disk $D_{j}$ (see Figure \ref{figure3}(a)) or (2)  outside of the disk $D_{j}$(see Figure \ref{figure3}(b)). 
	
	\begin{figure}[h!]
		
		\centering
		\includegraphics[scale=.9]{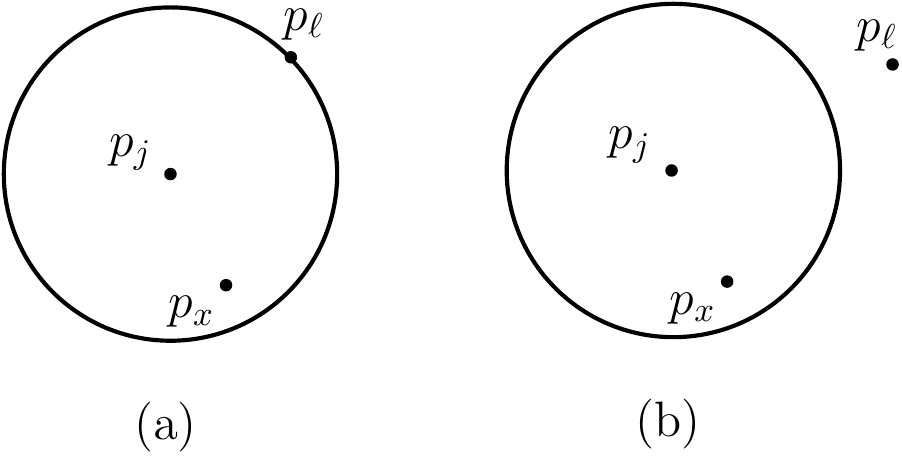}
		\caption{(a) $p_{\ell}$ lies on the boundary of the disk $D_{j}$ and  (b) $p_{\ell}$ lies outside of the disk $D_{j}$}\label{figure3}   			
		
	\end{figure}
	
	Since $d(p_j, p_{\ell})\geq \rho$ for both the above mentioned cases, therefore $cost_{2}(p_{j},S_{i}) \geq \rho$.
	Also, from  triangular inequality $d(p_{x}, p_{j}) +d(p_{x}, p_{\ell}) \geq d(p_{j},p_\ell) \geq \rho $ for each point $p_{\ell} \in S_{i}$. So, we can add the point $p_{j}$ to the set $S_{i}$ to construct set $S_{i+1}$. Here, $S_{i+1}=S_{i} \cup \{p_{j}\}$. Therefore,  $cost_{2}(S_{i+1}) \geq \rho$. 
	
	Now, if $D_j$ does not properly  contain any point from the set $S_{i}$, then  the distance of $p_{j}$ to any point of the set $S_{i}$ is greater than or equal to $\rho$. Since there exists at least one point $p_{j}  \in P \setminus S_i$ such that $cost_{2}(S_{i+1})=cost_{2}(S_{i} \cup \{p_{j}\}) \geq \rho$, therefore Algorithm \ref{algo} will always choose a point (see line number \ref{alphaline} of Algorithm \ref{algo}) in the iteration $i+1$ such that   $cost_{2}(S_{i+1})  \geq \rho $. 
	
	So, we can conclude that   $cost_{2}(S_{i+1})  \geq \rho $ and thus condition holds for $(i+1)$ too.
	
	 Therefore, Algorithm \ref{algo} produces a set $S_{k}$ of size $k$ such that $cost_{2}(S_{k}) \geq \rho$. Since $\rho \geq \frac{cost_{2}(S^*)}{2\sqrt 3}$, Algorithm \ref{algo} produces $2\sqrt3$-factor approximation result for the $2$-dispersion problem.
	
\end{proof}

\subsection{$2$-Dispersion Problem on a Line}

In this section, we discuss the $2$-dispersion problem on a line $L$. Let the point set $P=\{p_{1}, p_{2}, \ldots p_{n}\}$  be on a horizontal line arranged from left to right. Let $S_k \subseteq P$  be a solution returned by  Algorithm \ref{algo} and $S^*\subseteq P$  be an optimal solution. Note that, the value of $\gamma $ is $2$ and the value of $\lambda$ (line number \ref{line1} of Algorithm \ref{algo}) is $1$ in this problem. Let $s_{o}^*$ be a solution point and $s_{r}^*, s_{t}^*$ be supporting points, i.e., $cost_{2}(S^*)=d(s_{o}^*,s_{r}^*)+d(s_{o}^*,s_{t}^*)$. Let $S_{3}^*=\{s_{o}^*,s_{r}^*,s_{t}^*\}$.  We show that if $S_{3}=S_{3}^*$ in line number \ref{line2algo} of  Algorithm \ref{algo}, then   $cost_{2}(S_3)= cost_{2}(S^*)$. Let $S^*=\{s_{1}^*, s_{2}^*, \dots s_{k}^*\}$ are arranged from left to right.

\begin{lemma} \label{lemmaline}
	Let $S^*$ be an optimal solution. If $s_{o}^*$ is the  solution point and $s_{r}^*, s_{t}^*$ are supporting points, then both points $s_{r}^*$ and $s_{t}^*$ cannot be on the same side on the line $L$ with respect to $s_{o}^*$ and three points $s_{r}^*,s_{o}^*,s_{t}^*$ are consecutive  on the line $L$ in $S^*$.
\end{lemma}

\begin{proof}
	
	\begin{figure}[h!]
		
		\centering
		\includegraphics[scale=1]{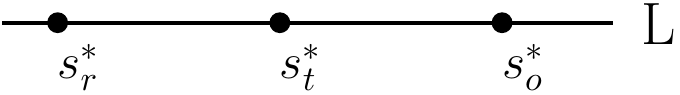}
		\caption{$s_{r}^*$ and $s_{t}^*$ on left side of $s_{o}^*$ }\label{figure4}   			
	\end{figure}
	
	On the contrary assume that   both $s_{r}^*$ and $s_{t}^*$ are on the left side of $s_{o}^*$,   and  $s_{t}^*$ lies between $s_{r}^*$ and $s_{o}^*$ (see Figure \ref{figure4}). Now, $d(s_{t}^*,s_{o}^*)+d(s_{t}^*,s_{r}^*) < d(s_{o}^*, s_{t}^*)+d(s_{o}^*, s_{r}^*) $ which leads to a  contradiction that $s_{o}^*$ is a solution point, i.e., $cost_{2}(S^*)=d(s_{o}^*, s_{r}^*)+d(s_{o}^*, s_{t}^*)$.  
	Now, suppose  $s_{r}^*,s_{o}^*,s_{t}^*$ are not consecutive  in $S^*$. Let $s^*$ be the point in $S^*$ such that either $s^* \in (s_{r}^*,s_{o}^*) $ or $s^* \in (s_{o}^*,s_{t}^*) $. If $s^* \in (s_{r}^*,s_{o}^*) $, then $d(s_{o}^*,s_{r}^*)+d(s_{o}^*+s_{t}^*) > d(s_{o}^*,s^*)+d(s_{o}^*+s_{t}^*)$, which leads to a contradiction that $s_{r}^*$ is a supporting point. Similarly, we can show that if $s^* \in (s_{o}^*,s_{t}^*) $, then $s_{t}^*$ is not a supporting point. Thus, $s_{r}^*,s_{o}^*,s_{t}^*$ are consecutive points on the line $L$ in $S^*$.
	
\end{proof}

Lemma \ref{lemmaline} says that if $s_{o}^*$ is a solution point, then $s_{o-1}^*$ and $s_{o+1}^*$ are supporting points as $s_{1}^*, s_{2}^*, \ldots , s_{k}^*$ are arranged from left to right.

\begin{lemma} \label{implemma_line}
	Let $S_{3}=\{s_{o}^*,s_{r}^*,s_{t}^*\}$	and $\alpha=cost_{2}(S_{3})$. Now, if  $S_{i}=S_{i-1} \cup \{p_{i}\}$ constructed in line number \ref{line10algo} of Algorithm \ref{algo}, then  $cost_{2}(S_{i}) = \alpha $. 
\end{lemma}

\begin{proof}
	We use induction to prove $cost_{2}(S_{i})=\alpha$ for $i=4,5 \ldots,k$.
	
	\textbf{Base Case:} Consider the set $S_{4}=S_{3} \cup \{p_{4}\}$ constructed in line number  \ref{line10algo} of Algorithm \ref{algo}. If $s_{o}^*$ is a solution points, and $s_{r}^*$, $s_{t}^*$ are supporting points and  $cost_{2}(p_{4},S_{3}) \geq \alpha$, therefore $p_{4} \notin [s_{o-1}^*,s_{o+1}^*]$ (otherwise one of $s_{o-1}^*$ and $s_{o+1}^*$ will not be supporting point). This implies  $p_{4}$ either lies
	in $[p_{1},s_{o-1}^*)$ or $(s_{o+1}^*,p_{n}]$. Assume $p_4 \in (s_{o+1}^*,p_{n}] $. In Algorithm \ref{algo}, we choose $p_{4}$ such that  $cost_{2}(p_{4},S_{4}) \geq \alpha$ (see line number \ref{alphaline} of Algorithm \ref{algo}) and $cost_{2}(p_{4},S_4)=\min_{q \in P \setminus S_{3}}cost_{2}(q,S_{4})$.  Therefore,  $p_{4} \in (s_{o+1}^*,s_{o+2}^*]$. Let   $S_{4}'=\{s_{1}^*,s_{2}^*,\ldots,s_{o-2}^*\} \cup S_{4} \cup \{s_{o+3}^*,s_{o+4}^*,\ldots,s_{k}^*\}$. Suppose $p_{4}=s_{o+2}^*$ and  we know that $S_{3}=S_{3}^*$ then $S_{4}'=S^*$.  
	So, $cost_{2}(S_{4}')=cost_{2}(S^*)=\alpha$. This implies $cost_{2}(S_4)=\alpha$. 
	Now assume that $p_{4} \in (s_{o+1}^*,s_{o+2}^*)$, then also we will show that $cost_{2}(S_{4}')=\alpha$.  We calculate  $cost_{2}(p_{4}, S_{4}')=d(p_{4},s_{o+1}^*) + d(p_{4},s_{o+3}^*) = d(s_{o+2}^*,s_{o+1}^*) +  d(s_{o+2}^*,s_{o+3}^*)\geq \alpha$ 
	and $cost_{2}(s_{o+3}^*, S_{4}')= d(s_{o+3}^*,p_4) + d(s_{o+3}^*,s_{o+4}^*) \geq d(s_{o+3}^*,s_{o+2}^*) + d(s_{o+3}^*,s_{o+4}^*)\geq \alpha$ (see Figure \ref{figure5}). Thus if $p_{4} \in (s_{o+1}^*,s_{o+2}^*)$, then $cost_{2}(S_4')=\alpha$. Therefore, if $k\geq 4$, then $p_{4}$ exists and  $cost_{2}(S_4)=\alpha$. Similarly, we can prove that if $p_4 \in [p_{1},s_{o-1}^*)$, then $cost_{2}(S_{4}')=\alpha$, where $S_{4}'= \{s_{1}^*,s_{2}^*,\ldots,s_{o-3}^*\} \cup S_{4} \cup \{s_{o+2}^*,s_{o+4}^*,\ldots,s_{k}^*\}$. In this case also $p_{4}$ exists and $cost_{2}(S_4)=\alpha$.
	\begin{figure}[h!]
		
		\centering
		\includegraphics[scale=0.9]{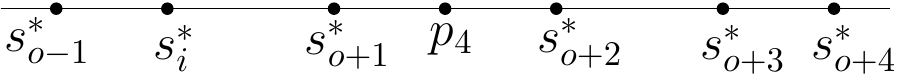}
		\caption{ Snippet of $S_{4}'$}\label{figure5}   			
		
	\end{figure}

	Now, assume that $S_{i}= S_{i-1} \cup \{p_i\}$ for $i<k$ such that $cost_{2}(S_{i}')=\alpha$  and $cost_{2}(S_{i})=\alpha$ where $S_{i}'=\{s_{1}^*,s_{2}^*,\ldots,s_{u}^*\} \cup S_{i} \cup \{s_{v}^*,s_{v+1}^*,\ldots,s_{k}^*\}$. If $p_{i} \in (s_{o}^*, p_{n}]$, then $s_{v-1}^* \in S^*$ is the left most point in the right of   $p_{i}$ and $u \geq  k-(i+k-v+1)=v-i-1$ with each point of $S_{i}$ are on the right side of $s_{u}^*$ (see Figure \ref{figure6}(a)) and if $p_{i} \in [p_{1}, s_{o}^*)$, then $s_{u+1}^*\in S^*$ is the right most point in the left of  $p_{i}$ where $v \geq u+i+1$ (see Figure \ref{figure6}(b)).  
	
	\begin{figure}[h!]
		
		\centering
		\includegraphics[scale=0.9]{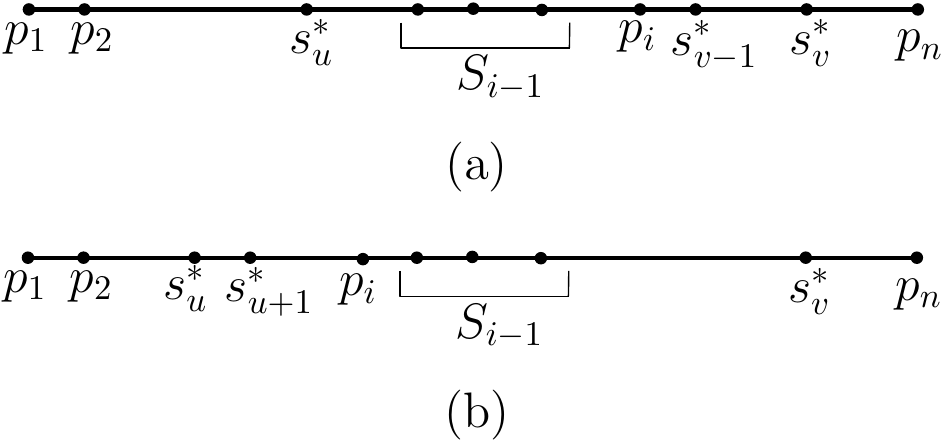}
		\caption{ Placement of  set $S_{i-1} \cup \{p_{i}\}$. }\label{figure6}   			
	\end{figure}
	
	We prove that $cost_{2}(S_{i+1})=\alpha$, where $S_{i+1}=S_{i} \cup \{p_i\}$.  It follows from the fact that size of $S_{i}$ is less than $k$, and the set $\{s_{1}^*,s_{2}^*,\ldots,s_{s_{u}}^*\} \cup \{s_{v}^*,s_{v+1}^*,\ldots,s_{k}^*\} \neq \phi$ and  the  similar arguments discussed in the base case.
\end{proof}
\begin{lemma}\label{prev_lemma}
	The running time of Algorithm \ref{algo} on  line  is $O(n^4)$. 
\end{lemma}

\begin{proof}
	Since it is a $2$-dispersion problem on a line, so  algorithm starts by setting $\lambda=1$ in line number \ref{line1} of Algorithm \ref{algo}, and then compute solution set for each distinct $S_{3} \subseteq P$ independently. Now, for each $S_{3}$,  algorithm selects a point iteratively based  on  greedy choice (see line number \ref{alphaline} of Algorithm \ref{algo}).  Now, for choosing remaining $(k-3)$ points, the total amortize time taken by the algorithm is $O(n)$.  So, the overall time complexity of Algorithm \ref{algo} on  line consisting of $n$ points is  $O(n^4)$. 
\end{proof}

\begin{theorem}
	Algorithm \ref{algo} produces an optimal solution for the $2$-dispersion problem on a line in polynomial time. 
\end{theorem} 
\begin{proof}
	Follows from  Lemma \ref{implemma_line}  that $cost_{2}(S_{i}) = \alpha=cost_{2}(S_{3}^*)$ for $3 \leq i \leq k$, where $S_{3}=\{s_{o}^*,s_{r}^*,s_{t}^*\}$. Therefore, $cost_{2}(S_k) = \alpha$. Also,  Lemma \ref{prev_lemma} says that Algorithm \ref{algo} computes $S_{k}$ in polynomial time. Thus, the theorem. 
\end{proof}

\subsection{$1$-Dispersion Problem in $\mathbb{R}^2$} 
In this section, we show the effectiveness of Algorithm \ref{algo} by  showing $2$-factor approximation result for the $1$-dispersion problem in $\mathbb{R}^2$. Here, we set $\gamma=1$ as input along with input $P$ and $k$. We also set $\lambda=2$ in line number \ref{line1} of the algorithm \ref{algo}.

Let $S^{*} $  be an optimal solution for a given instance $(P,k)$ of $1$-dispersion problem and  $S_{k} \subseteq P$ be a solution returned by our greedy Algorithm \ref{algo} provided $\gamma=1$ as an additional input. Let $s_{o}^{*} \in S^{*}$  a solution point, i.e.,  $cost_{1}(S^{*}) = d(s_{o}^{*}, s_{r}^{*})$  such that     $s_{r}^{*}$ is the closest points  of $s_{o}^{*}$ in $S^{*}$. We call $s_{r}^{*}$ as supporting point. Let $\alpha=d(s_{o}^*, s_{r}^*)$ and $\rho=\frac{\alpha}{2}$.

We define a disk  $D_{i}$ centered at  $p_{i} \in P$ as follows: $D_{i} = \{p_j \in \mathbb{R}^2|d(p_{i}, p_{j}) \leq \rho\}$. Let $D=\{D_{i}\mid p_{i} \in P\}$. Let $D^*$ be the subsets of $D$ corresponding  to disks centered at  points in $S^*$. If $d(p_{i},p_{j}) <  \rho$, then we say that $p_{j}$ is \emph{properly contained} in $D_{i}$ and if $d(p_{i},p_{j}) \leq  \rho$, then we say that  $p_{j}$ is \emph{ contained} in $D_{i}$.

\begin{lemma} \label{lemma1_1disp}
	For any point $s \in P $, if $D^s=\{q \in \mathbb{R}^2 \mid d(s,q) \leq \rho\}$ then $D^s$ properly contains at most one point of the optimal set $S^*$.
\end{lemma}
\begin{proof}
	On the contrary assume that  $p_{a},p_{b} \in S^*$ such that $p_{a},p_{b}$ are properly contained in $D^s$. If two points $p_{a} $ and $p_{b}$ are properly contained in $D^s$, then $d(p_{a},p_{b})< d(p_{a},s)+d(p_{b},s) <\frac{\alpha}{2}+\frac{\alpha}{2}=\alpha$, which leads to a contradiction to the optimality of $S^*$. Thus, the lemma.
\end{proof}	
\begin{lemma}
	For any  two points $p_{a}, p_{b} \in S^*$, there does not exist any point  $s \in \mathbb{R}^2$ that is properly contained in  $D_a \cap D_b$. \label{lemma2_1disp}
\end{lemma}
\begin{proof}
	On the contrary assume that $s$ is properly contained in   $D_{a} \cap D_{b} $. This implies  $d(p_{a}, s) < \frac{\alpha}{2}$ and
	$d(p_{b}, s) <  \frac{\alpha}{2}$ . Therefore, the disk $D^s =\{q \in \mathbb{R}^2 \mid d(s,q) \leq \rho \}$ properly contains two points $p_{a}$ and $p_{b}$, which is a contradiction to Lemma \ref{lemma1_1disp}. Thus, the lemma.  
	
\end{proof}

\begin{corollary} \label{cor11}
	For any point $s \in P $, if $D' \subseteq D^*$ is the set of disks that contain  $s$, then $|D'| \leq 2$ and $s$ lies on the boundary of both the  disk in $D'$.  
\end{corollary} 
\begin{proof}
	Follows from   Lemma \ref{lemma2_1disp}.  
\end{proof}
\begin{corollary} \label{cor12}
	For any point $s \in P $, if $D'' \subseteq D^*$ be a subset of disks that properly contains  point $s$, then $|D''| \leq  1$.
\end{corollary}
\begin{proof}
	Follows from Lemma \ref{lemma2_1disp} and Corollary \ref{cor11}. 
\end{proof}

\begin{lemma} \label{lemma1displast}
	Let $M\subseteq P$ be a set of points such that $|M| < k$. If $cost_{2}(M)\geq \frac{\alpha}{2}$, then there exists at least one disk $D_{j} \in D^*=\{D_{1},D_{2}, \ldots, D_{k}\}$ that does not properly contain any  point from the set $M$.
\end{lemma}
\begin{proof}
	On the  contrary, assume that each   $D_{j} \in D^*$ properly contains at least one point from the set $M$. Construct a bipartite graph $G(M \cup  D^*, {\cal E})$ as follows: (i) $M$ and $D^*$ are two partite vertex sets, and (ii) for $u\in M$, $(u,  D_j) \in {\cal E}$ if and only if $u$  is properly contained in $ D_j$. According to assumption, each disk $ D_j$ contains at least $1$ points from the set $M$. Therefore, the total degree of the vertices in $ D^*$ in $G$ is at least $k$. Note that $| D^*| = k$. On the other hand, the total degree of the vertices in $M$ in   $G$ is at most $|M|$ (see Corollary  \ref{cor12}). Since $|M| < k$, the total degree of the vertices in $M$ in   $G$ is less than $k$, which leads to a contradiction that the total degree of the vertices in $ D^*$ in   $G$ is at least $k$. Thus, there exist at least one disk $D_{j} \in D^*$ such that $D_{j}$ does not properly  contain any point from the set $M$.  
\end{proof}

\begin{theorem}
	Algorithm \ref{algo} produces a $2$-factor approximation result for the $1$-dispersion problem in $\mathbb{R}^2$.
\end{theorem}
\begin{proof}

	Since it is a $1$-dispersion problem, so $\gamma=1$ and set $\lambda=2$ in line number \ref{line1} of the algorithm.
	Now, assume $\alpha=cost_{1}(S^*)$ and $\rho=\frac{\alpha}{\lambda}=\frac{cost_{1}(S^*)}{2}$, where $S^*$ is the optimum solution.
	Here, we show that   Algorithm \ref{algo} returns a solution  set $S_{k}$ of size $k$ such that $cost_{1}(S_{k}) \geq \rho$. More precisely, we show that Algorithm \ref{algo} returns a solution $S_{k}$ of size $k$ such that  $cost_{1}(S_{k})\geq \rho$ and $S_{k} \supseteq \{s_{o}^{*},s_{r}^{*}\}$, where $s_{o}^*$ is the solution point and $s_{r}^*$ is the supporting point. Our objective is to show that if $S_{2}=\{ s_{o}^{*},s_{r}^{*}\}$  in  line number \ref{line2algo}  of  Algorithm \ref{algo}, then it computes a solution $S_{k}$ of size $k$ such that $cost_{1}(S_{k}) \geq \rho$. Note that  any other solution returned by  Algorithm \ref{algo} has a $1$-dispersion cost better than $\frac{cost_{1}(S^*)}{2}$. Therefore, it is sufficient to prove that if  $S_{2}=\{ s_{o}^{*},s_{r}^{*}\}$ in line number \ref{line2algo} of  Algorithm \ref{algo}, then the size of  $S_{k}$ (updated) in line number \ref{line15algo}  of  Algorithm \ref{algo} is $k$ as every time Algorithm \ref{algo} added a point (see line number \ref{line10algo}) into the set  with the property that $1$-dispersion  cost of the updated set is greater than or equal to  $ \rho= \frac{cost_{1}(S^*)}{2}$. Therefore,  we consider $S_{2}=\{ s_{o}^{*},s_{r}^{*}\}$ in  line number \ref{line2algo}  of  Algorithm \ref{algo}.

	We use induction to  establish the condition    $cost_{1}(S_{i}) \geq \rho$  for each $i=3,4,\ldots k$. Since $S_{2}=S_{2}^*$, therefore $cost_{1}(S_{2})=cost_{1}(S_{2}^*)=\alpha$ holds. Now, assume that the condition $cost_{1}(S_{i})\geq \rho$ holds for each $i$ such that  $3\leq i < k$. We will prove that the condition $cost_{1}(S_{i+1})\geq \rho$ holds for $(i+1)$ too.

	Let $D^*$ be the set of disks centered at the points in $S^*$ such that the radius of each disk be $\rho=\frac{\alpha}{2}$. Since $i<k$ and $S_{i}\subseteq P$ with condition $cost_{1}(S_{i})\geq \rho=\frac{\alpha}{2}$, there exist at least one disk, say $D_{j} \in D^*$ that does not contain any  point from the set  $S_{i}$ (see Lemma \ref{lemma1displast}). We will show that $cost_{1}(S_{i+1})=cost_{1}(S_{i} \cup \{p_{j}\}) \geq \rho$, where $p_{j}$ is the center of the disk $D_{j}$.
	
	Now, if $D_j$ does not properly  contain any point from the set $S_{i}$, then  the closest point  of $p_{j} \in S^*$ may lie on the boundary of the disk $D_{j}$ (by Corollary \ref{cor11}) or outside the disk $D_{j}$ (by Lemma \ref{lemma1displast}). In both the cases, distance of $p_{j}$ to any point of the set $S_{i}$ is greater than or equal to $\rho$ (see Figure \ref{figure7}). Since there exists at least one point $p_{j}  \in P \setminus S_i$ such that $cost_{1}(S_{i+1})=cost_{1}(S_{i} \cup \{p_{j}\}) \geq \rho$, therefore Algorithm \ref{algo} will always choose a point (see line number \ref{alphaline}  of Algorithm \ref{algo}) in the iteration $i+1$ such that  $cost_{1}(S_{i+1})  \geq \rho $.
	
		So, we can conclude that   $cost_{1}(S_{i+1})  \geq \rho $ and thus condition holds for $(i+1)$ too.

	\begin{figure}[h!]
		
		\centering
		\includegraphics[scale=1]{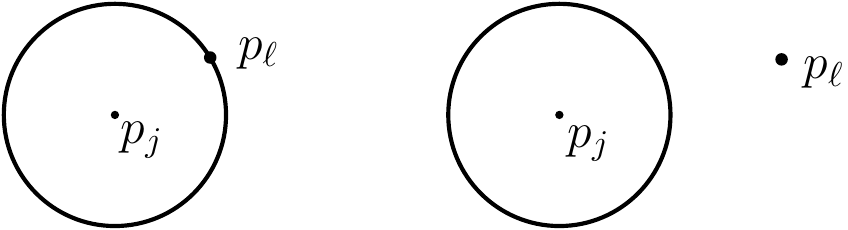}
		\caption{ Second closest point of $p_{j}$ lies on boundary of $D_{j}$ or outside of $D_{j}$. }\label{figure7}   			
	\end{figure}
Therefore, 	Algorithm \ref{algo} produces a $2$-factor approximation result for the $1$-dispersion problem in $\mathbb{R}^2$.
\end{proof}

\section{Conclusion } \label{section5}
In this article, we proposed a $(2\sqrt3+\epsilon)$-factor approximation algorithm for the $2$-dispersion problem in $\mathbb{R}^2$, where $\epsilon > 0$. The best known approximation factor available in the literature is $4\sqrt3$ \cite{amano2020}. Next, we proposed a common framework for the dispersion problem. Using the framework, we  further improved the approximation factor  to  $2\sqrt{3}$ for the $2$-dispersion problem in $\mathbb{R}^2$. We studied the $2$-dispersion problem on a line  and proposed a polynomial time algorithm that returns an optimal solution using the developed framework. Note that, for the $2$-dispersion problem on a line,  one can propose a polynomial time algorithm  that returns an optimal value in relatively low time complexity, but to show the adaptability and flexibility of our proposed framework, we presented an algorithm for the same problem using the developed framework. We also proposed a $2$-factor approximation algorithm for the $1$-dispersion problem using the proposed common framework to show effectiveness of the framework.


\small
\bibliographystyle{abbrv}


\end{document}